\documentclass[runningheads]{llncs}

\usepackage{amsmath}
\usepackage{amssymb}
\usepackage{rotating}

\usepackage{algorithmic}
\usepackage{algorithm}
\usepackage{booktabs}

\usepackage{tikz}
\usepackage{pgfplots}
\usetikzlibrary{arrows}

\definecolor{plotcolor1}{rgb}{0,0,.8}
\definecolor{plotcolor2}{rgb}{.8,0,0}
\definecolor{plotcolor3}{rgb}{0,.6,0}
\definecolor{plotcolor4}{rgb}{0,.8,.8}

\colorlet{maxcolor1}{black!50!plotcolor1}
\colorlet{mincolor1}{white!50!plotcolor1}
\colorlet{maxcolor2}{black!50!plotcolor2}
\colorlet{mincolor2}{white!50!plotcolor2}

\def\linestyleA{solid}
\def\linestyleB{dashed}
\def\linestyleC{dashdotted}
\def\linestyleD{densely dotted}
\newlength{\figurewidth}
\newlength{\subfigheight}
\newlength{\subfigtop}
\pgfplotsset{
compat=1.14,
width=.85\figurewidth, height=0.35\figurewidth,
scale only axis,
every axis plot/.append style={line width=1pt},
every axis plot/.append style={mark size=2pt},
grid style=dashed,
label style={font=\small},
legend style={font=\scriptsize, inner xsep=2pt, inner ysep=1pt, nodes={inner sep=1pt, text depth=1pt}},
major tick length=3pt,
minor tick num=4,
minor tick length=1pt,
scaled ticks=false,
tick label style={font=\scriptsize},
title style={yshift=-7pt},
xmajorgrids,
ymajorgrids,
/pgf/number format/.cd,
set thousands separator={\,},
/tikz/.cd
}

\usepackage{tkz-berge}

\newcommand{\DCOP}{\mathcal{T}}
\newcommand{\var}[1]{X_{#1}}
\newcommand{\allvars}{\mathcal{X}}
\newcommand{\agent}[1]{A_{#1}}
\newcommand{\allagents}{\mathcal{A}}
\newcommand{\domain}[1]{D_{#1}}
\newcommand{\alldomains}{\mathcal{D}}
\newcommand{\constraint}{C}
\newcommand{\constraintsize}{k}
\newcommand{\allconstraints}{\mathcal{R}}
\newcommand{\varvalue}{\upsilon}

\newcommand{\numvars}{n}
\newcommand{\numagents}{\numvars}
\newcommand{\somevars}[1]{\allvars_{#1}}
\newcommand{\neighborset}[1]{\mathcal{M}_{#1}}
\newcommand{\neighborsize}[1]{|\neighborset{#1}|}
\newcommand{\domainsize}[1]{|\domain{#1}|}
\newcommand{\cpa}[1]{\hat{\mathcal{X}}_{#1}}
\newcommand{\posreal}[1]{{\mathbb R}^{{#1}}_{\geq0}}

\DeclareMathOperator*{\argmin}{arg\,min}

\newlength{\graphwidth}
\begin{document}

\title{Hybrid DCOP Solvers: Boosting Performance of Local Search Algorithms}

\pdfinfo{
/Title (Hybrid DCOP Solvers: Boosting Performance of Local Search Algorithms)
/Author (Cornelis Jan van Leeuwen and Przemyz{\l}aw Pawe{\l}czak)
/Abstract (We propose a novel method for expediting both symmetric and asymmetric Distributed Constraint Optimization Problem (DCOP) solvers. The core idea is based on initializing DCOP solvers with greedy fast non-iterative DCOP solvers. This is contrary to existing methods where initialization is always achieved using a random value assignment. We empirically show that changing the starting conditions of existing DCOP solvers not only reduces the algorithm convergence time by up to 50\%, but also reduces the communication overhead and leads to a better solution quality. We show that this effect is due to structural improvements in the variable assignment, which is caused by the spreading pattern of DCOP algorithm activation.)
/Subject (Hybrid DCOPs)
}

\author{Cornelis Jan van Leeuwen\inst{1,2} \and
Przemyz{\l}aw Pawe{\l}czak\inst{2}}

\institute{TNO, Eemsgolaan 3, Groningen, The Netherlands \\
\email{coen.vanleeuwen@tno.nl}
\and
Delft University of Technology, Van Mourik Broekmanweg 6, Delft, The Netherlands \\
\email{\{p.pawelczak,c.j.vanleeuwen-2\}@tudelft.nl}}

\maketitle

\begin{abstract}
We propose a novel method for expediting both symmetric and asymmetric Distributed Constraint Optimization Problem (DCOP) solvers. The core idea is based on initializing DCOP solvers with greedy fast non-iterative DCOP solvers. This is contrary to existing methods where initialization is always achieved using a random value assignment. We empirically show that changing the starting conditions of existing DCOP solvers not only reduces the algorithm convergence time by up to 50\%, but also reduces the communication overhead and leads to a better solution quality. We show that this effect is due to structural improvements in the variable assignment, which is caused by the spreading pattern of DCOP algorithm activation.
\end{abstract}

\section{Introduction}
\label{sec:introduction}
Distributed Constraint Optimization Problems (DCOPs) are a method for formalizing and solving problems that have a distributed nature, and in which multiple cooperating agents control discrete variables in order to optimize a common problem. DCOPs can be found in many different domains such as sensor networks~\cite{Farinelli2014}, mobile sensing team coordination~\cite{Yedidsion2014}, communication~\cite{Yeoh2012}, home automation~\cite{Rust2016} and smart grid optimization~\cite{Fioretto2017}. The underlying structure of DCOPs is always the same: agents have to assign variables that not only optimize a local set of constraints, but have to send messages to other agents in order to cooperatively come up with variable assignments that are optimal for the complete set of agents. A special kind of DCOPs can be formulated where agents having a shared constraint may assign different costs for a value assignment. These problems are called Asymmetric DCOPs (ADCOPs)~\cite{Grinshpoun2013}. In ADCOPs the aspect of cooperation is even more important, since an assignment leading to a local improvement may deteriorate the global performance.

A variety of algorithms that find solutions for DCOPs can be classified as either \emph{complete} or \emph{incomplete} solvers~\cite{Leite2014}. Solvers such as ADOPT~\cite{Modi2005}, DPOP~\cite{Petcu2005} or AFB~\cite{Gershman2009} are the algorithms of the complete type, and are used to find the optimal solution. However, DCOPs are NP-hard~\cite{Modi2003} so a solution becomes intractable for large scale problems. Therefore incomplete solvers use heuristic approaches to find a solution which may be suboptimal, but reaching the solutions much faster. In other words, there is a trade-off between solution quality and speed.

In this paper we introduce a new class of incomplete DCOP solvers, that combine features of different DCOP solvers and combine them into \emph{hybrid solvers}. Specifically, we show that we can use different \emph{initialization} methods for existing DCOP algorithms, which has a profound impact on their overall performance. In order to do so we show that using different initialization methods that are \emph{not} iterative in approach, and are hence very fast in converging to a solution, we can reduce algorithm running times and improve the solution quality.

In the evaluation of the proposed hybrid DCOP solvers, we will consider existing symmetric and asymmetric solvers, as well as a new algorithm which is an extension of ACLS, which we will refer to as ACLS-UB.

\section{Problem Statement and Notation}
\label{sec:problem_statement}

Before introducing our proposed novel class of hybrid DCOP solver we need to start with providing the definition of (A)DCOPs. DCOPs are problems from the field of multi-agent systems in which agents can reason and send messages to one another to cooperatively decide on their variable assignments in order to find a solution to a global cost minimization function.

\paragraph{Problem Formalization and Notation}

Following the notation from~\cite{Gershman2009}, DCOPs are defined as a tuple $\DCOP=\langle \allagents, \allvars, \alldomains, \allconstraints \rangle$ where $\allagents$ is a finite set of agents $\{\agent{1}, \agent{2}, \ldots, \agent{\numagents}\}$ and $\allvars$ is the set of variables $\{\var{1}, \var{2}, \ldots, \var{\numvars}\}$ with finite discrete domains $\{\domain{1},\domain{2},\ldots,\domain{\numvars}\}$ from the set of domains $\alldomains$ such that $\var{i} \in \domain{i}$. Furthermore we require that each agent $\agent{i}$ is assigned one corresponding variable, $\var{i}$, and therefore $|\allagents| = |\allvars| = |\alldomains|$. This one-on-one relation between agents and variables, is seen in many DCOP studies, but is not strictly required. Then, $\allconstraints$ is a set of relations or constraints between variables, in which each constraint $\constraint \in \allconstraints$ defines a non-negative cost depending on the value assignment of the involved variables. Every possible value assignment of a set of variables has a particular induced global cost $\constraint{:~} \domain{i_1} \times \domain{i_2} \times \ldots \times \domain{i_\constraintsize} \to \posreal{}$, while for ADCOPs each constraint defines a set of costs for every involved variable, i.e. $\constraint{:~} \domain{i_1} \times \domain{i_2} \times \ldots \times \domain{i_\constraintsize} \to \posreal{\constraintsize}$. Having all definitions of $\DCOP$, in (A)DCOPs the goal of the agents is to minimize the global cost function, i.e.
\begin{equation}
\label{eq:min}
\argmin_\allvars \sum{\allconstraints}.
\end{equation}
In the rest of this paper, as for most DCOP studies, we shall only take into account binary constraints, in which exactly two variables are involved for every constraint, which is then of the form $\constraint_{i,j}{:~}\domain{i} \times \domain{j} \to \posreal{2}$.

\paragraph{Definitions}

We refer to agents as \emph{neighbors} if there exists a constraint between the corresponding variables. This follows the real-life situation of limited range between cooperative agents, e.g. communication range in wireless networks. The set of all neighbors of an agent $\neighborset{i} \subseteq \allagents$ is called the \emph{neighborhood}. The set $\cpa{i}$ denotes the set of known assigned values of the neighbors of $\agent{i}$ and is also referred to as the \emph{current partial assignment} (CPA). Note that the constraints between variables can be depicted as an undirected graph.

\section{A New Class of DCOP Solvers: Hybrid Solvers}
In this paper we propose a new class of DCOP solvers in which we propose a simple, yet an effective idea. We take the best of different types of existing DCOP solvers, forming a new class of \emph{hybrid} DCOP solvers. Particularly, we propose to improve the performance of existing local search algorithms by modifying the initialization methods to find the initial value assignment. In the field of Constraint Satisfaction Problems, which is closely related to DCOPs, the approach of using a initialization and repairing it is well known, and can yield great benefits~\cite{Minton1992}. From the fields of evolutionary algorithms~\cite{Rahnamayan2007}, clustering~\cite{Cao2009}, neural networks~\cite{Yam2000,Dolezel2016} and meta-learning~\cite{Feurer2015} we know that initialization methods can have a great effect on the performance of an algorithm. However to the best of our knowledge, there has been little to no work on the effects of different initialization methods for (A)DCOPs. In this paper we will study the effect different initialization methods may have on the performance of (existing) DCOP algorithms.

\textbf{Hybrid DCOP Solver:} We define a hybrid DCOP solver as a solver which executes sequentially other (existing) DCOP solvers. Selection of \emph{which} DCOP method to use in the next DCOP solving iteration and \emph{when} to switch to a new DCOP solver method is a core of the hybrid solver definition.

\subsection{Motivation for a Hybrid DCOP Solver}

Most (if not all) local search DCOP algorithms use an initial random assignment for all of the variables, which will be iteratively improved upon. Instead, an initial assignment can be computed by a non-iterative DCOP algorithm, such as a simple greedy algorithm, or a more elaborate greedy algorithm such as the one introduced in~\cite{vanLeeuwen2017aaai}. Since these methods will assign a value only once, and then terminate, they quickly provide a good initialization assignment from which one can start another DCOP method.

We hypothesize that the combination of different initialization methods for iterative algorithms in DCOP solution search, will be beneficial because of two effects:

\begin{enumerate}
\item \emph{Solution quality improvement over initial assignment}: most probably a simple initialization method will find a sub-optimal solution, and many local search algorithms will be able to improve it. Algorithms that are known to provide monotonically decreasing solution costs (any algorithm that uses a coordinated change approach, e.g. MGM-2, ACLS, MCS-MGM) are guaranteed to find better (or equal) solutions compared with the initial value assignment; and
\item \emph{Increased convergence speed for local search algorithms}: DCOP algorithms that use local search will most likely converge faster when a good solution is used for initial DCOP value assignment. This will lead to a shorter total running time for algorithms that are initialized with a better assignment.
\end{enumerate}

\section{Initialization of DCOP Solvers: Classification}

In our experiments we combine different \emph{initialization methods} with existing \emph{local search algorithms} which we shall also refer to as \emph{iterative methods}. Since the aim of this study is to improve the solution quality and convergence speed of solvers, we do not take into account complete solvers. To understand why only certain combination of DCOP solvers improves the solution, we need first to classify (i) initialization methods, and (ii) types of DCOP iterative solvers.

\subsection{DCOP Classification: Initialization Methods}

\begin{itemize}
\item\textbf{Random} A de facto standard method for all DCOP solvers. It does not take into account any constraints and starts with the random variable assignment.

\item \textbf{$k$ Step Look-ahead:} We define a look-ahead initialization algorithm as the one in which one randomly chosen initial node is triggered first, and only after it has chosen a value it will activate its neighbors. When choosing a value, it takes into consideration the effect on all of its neighbors that are reachable within $k$ steps (edges or hops). Three special cases of $k$ step look-ahead are already known and described in the literature:

\begin{itemize}
\item \textbf{Zero Step Look-ahead (ZSLA):} A zero step look-ahead algorithm ($k=0$) is the one in which an agent optimizes only for the constraints it is directly involved in. Such algorithm are also referred to as \emph{greedy, breadth-first} algorithms;

\item \textbf{Single Step Look-ahead (SSLA):} A single-step-look-ahead algorithm ($k=1$) is defined as the one in which an agent optimizes not only for the constraints it is involved in, but also the constraints its one-hop neighbors are in. One such algorithm is the recently proposed CoCoA algorithm~\cite{vanLeeuwen2017aaai} and its variants CoCoA\_UF and CoCoA\_WPT~\cite{vanLeeuwen2017saso};

\item \textbf{Max Step Look-ahead (MSLA):} If $k$ is equal to the height of the graph's minimal spanning tree, and the algorithm would be started at the root of the spanning tree, the algorithm becomes a complete algorithm, and is in fact equivalent to DPOP~\cite{Petcu2005}.
\end{itemize}
\end{itemize}

\subsection{DCOP Classification: Existing Iterative Methods}
Classifying iterative methods used in DCOP solvers we can divide them into two main groups:

\begin{itemize}
\item\textbf{Symmetric DCOP Solvers:} Those include DSA~\cite{Zhang2005}, MGM and MGM2~\cite{Maheswaran2004} and generalized DBA~\cite{Okamoto2016};

\item\textbf{Asymmetric DCOP Solvers:} Those include MCS-MGM~\cite{Grubshtein2010}, and ACLS~\cite{Grinshpoun2013} with its new version ACLS-UB (which is also a novel contribution of this work and described in Section~\ref{sec:aclsub}).

\end{itemize}

\paragraph{Remark on Max-Sum:} We are naturally aware of another popular DCOP solver: the Max-Sum algorithm~\cite{Farinelli2008} or any one of its variants. However Max-Sum is unable to utilize the benefit of initializing, as it tries to approximate the global utility of any value, and uses this to determine the best variable assignment. There are extensions of Max-Sum that are able to build upon an an initial assignment by using value propagation~\cite{Zivan2012}. In a recent paper~\cite{Chen2017} the effect of initialization was studied in variant called Max-Sum\_ADSSVP. The authors find that the timing, and approach to initialization has a great effect on the performance, both in terms of solution quality and convergence speed. For our evaluation however, we leave Max-Sum out of the comparison, and refer to this paper to provide a complete overview of different hybrid algorithms.

\subsection{Novel Iterative DCOP Solver: ACLS-UB}
\label{sec:aclsub}
In addition to existing DCOP algorithms listed above we introduce a variant of the ACLS, denoted as \emph{Unbiased} (ACLS-UB).

\paragraph{ACLS-UB Algorithm:} In the original ACLS algorithm~\cite{Grinshpoun2013}, at every iteration an agent chooses a variable assignment that would lower its local costs and proposes it as a new value to its neighbors. Neighbors respond with the effect on their side, after which the proposition which has the best effect on the regional cost function is selected. In the ACLS-UB algorithm a value assignment is proposed from \emph{all} possible values, instead from the subset that improves its local state. The ACLS-UB algorithm is described using pseudo code in Algorithm~\ref{algo:aclsub}.

ACLS-UB works by iteratively proposing a random value from its domain $\domain{i}$, and sends that value to its neighbors. The neighbors respond by sending the effect of the assignment on their local costs, taking into account all known value assignments. When these local effects are received by the initial agent, it sums over all received effects, and assigns the proposed value with probability $p$ only if it will reduce the current local cost.

\renewcommand{\algorithmicensure}{}
\begin{algorithm}[t]
\caption{\label{algo:aclsub}ACLS-UB Algorithm}
\begin{algorithmic}[1]
\ENSURE On $\agent{i}$ when activated
\STATE $\var{i} \leftarrow \text{chooseRandomValue()}$
\WHILE {(no termination condition is met)}
\STATE send $\var{i}$ to $\forall \agent{j} \in \neighborset{i}$
\STATE \label{line:choose} $\varvalue \leftarrow \text{chooseRandomValue()}$
\STATE send $\varvalue$ to $\forall \agent{j} \in \neighborset{i}$
\STATE wait for incoming constraint cost $\delta_{j}$ from $\agent{j}$
\STATE $\Delta_{i} \leftarrow \sum_{j}^{j\in\neighborset{i}}{\delta{j}}$
\IF {$\Delta_{i} < $ currentCost \AND random$[0,1]$ $< p$}
\STATE assign $\var{i} \leftarrow \varvalue$
\ENDIF
\ENDWHILE

\medskip
\ENSURE On $\agent{j}$ when receiving $\varvalue$ from $\agent{i}$
\STATE send constraint cost $\delta_{j}$ for $\var{i} = \varvalue$
\end{algorithmic}
\end{algorithm}

\paragraph{Relation of ACLS-UB to Other Solvers:} The main difference between ACLS and ACLS-UB is in line~\ref{line:choose} of Algorithm~\ref{algo:aclsub}, where any random value is picked from the domain. In the long run, the effect of this pick is that the effect of all values from the domain are used to retrieve the induced effect on the neighbors' local cost.

Intuitively ACLS-UB works very similar to CoCoA~\cite{vanLeeuwen2017aaai,vanLeeuwen2017saso}, with the major difference that CoCoA operates in one single iteration instead of iteratively trying different values. Another difference is that in ACLS-UB the neighbors will send back the value of the constraint cost, whereas CoCoA will send back the lowest induced cost for any assignment in conjunction with the proposed value and the CPA. This extra look-ahead is not efficient in ACLS-UB, since the next-hop neighbors \emph{will} in fact already have an assignment, and the lowest cost will be too optimistic. Note that the unique-first approach of CoCoA is not required in ACLS-UB, as it can easily recover from any earlier suboptimal assignments in later iterations, whereas CoCoA cannot.

\section{Hybrid DCOP Solvers: Introduction and Initial Results}
\label{sec:first_results}

In order to understand whether there is any benefit from hybrid solvers, we performed the following experiments\footnote{For reproducibility and validation of our results, all (Java) code for the algorithms is available at \mbox{\url{https://github.com/coenvl/jSAM/tree/OptMAS18}}, and for the experimental setups (MATLAB) at \mbox{\url{https://github.com/coenvl/mSAM/tree/OptMAS18}}.}. For any problem, we initiate 200 problem instances which are initialized by three methods: \emph{random}, \emph{ZSLA} (i.e. greedy) and \emph{SSLA} (i.e. CoCoA) and subsequently solved by other DCOP solvers (depending on the experiment). We report the average result of all problems. We assume a solver has converged when no better solutions have been found for more than 100 iterations, and define the moment of ``convergence'' as the first iteration in which the solution was within 1\% of the minimal solution. In this way we can compare the convergence speed of different algorithms, and do not have to specify the number of iterations beforehand, in a way similar to the any-time solution as proposed in~\cite{Zivan2014}.

As performance metrics we will score solvers on the following metrics:
\begin{itemize}
\item \textbf{The number of iterations required to converge}, denoted as {I};
\item \textbf{Final cost of the solution after the algorithm converged}, denoted as {S};
\item \textbf{Number of messages that are transmitted during the run}, denoted as {M};
\item \textbf{Number of constraint evaluations}, denoted as {E}; and
\item \textbf{Running time until the moment of convergence}, denoted as {T} in seconds.
\end{itemize}
The constraint evaluations are indicative of the computational complexity and can also be referred to as Non Concurrent Constraint Checks (NCCC)s~\cite{Meisels2002}.

\subsection{Experiment Results}
\label{sec:experiment_init}

\begin{figure}[t]
\begin{center}
\begin{tikzpicture}
\begin{axis}[
xmin=0,
xmax=2,
xlabel={Running time (s)},
ymin=50,
ymax=200,
ylabel={Solution cost},
x tick label style={
/pgf/number format/.cd,
fixed,
fixed zerofill,
precision=1,
/tikz/.cd
},
legend style={legend cell align=left, align=left, draw=black}
]
\addplot [color=plotcolor1, style=\linestyleA] table[]{
0.0509758612025529 194.03
0.09549031738582 194.03
0.119534556956994 167.835
0.159732878877887 167.835
0.180948527140603 137.84
0.220252570670885 137.84
0.256811044339492 137.84
0.279596230448515 129.16
0.319596764694406 129.16
0.340616426766558 105.055
0.38001131762541 105.055
0.412023655970444 105.055
0.433763348397974 101.81
0.474005967143695 101.81
0.49468925305131 87.67
0.533850304666027 87.67
0.563773984084942 87.67
0.585097952981256 86.47
0.625686398391355 86.47
0.646651003962901 79.125
0.686243689790214 79.125
0.715386640862408 79.125
0.736732362500971 78.63
0.777373056429859 78.63
0.798724149701169 75.205
0.839039046628179 75.205
0.867582065091935 75.205
0.889040264635876 75.025
0.929597559681465 75.025
0.950399893077887 73.79
0.989530168117883 73.79
1.01818184453773 73.79
1.03965955226183 73.69
1.08062135532166 73.69
1.10163142284815 73.265
1.14162274545499 73.265
1.16992153696526 73.265
1.19090421755222 73.235
1.23104637071997 73.235
1.25160042746965 73.02
1.29077792948173 73.02
1.31893118912924 73.02
1.33965421709638 73.02
1.37935604771819 73.02
1.3999592879104 72.85
1.43903408525252 72.85
1.46705040290892 72.85
1.48803214628637 72.84
1.52788186273491 72.84
1.54914528497916 72.71
1.58863750705186 72.71
1.61681377938671 72.71
1.63774420963152 72.695
1.67778760389077 72.695
1.6985663520633 72.605
1.73830990491882 72.605
1.76642688872333 72.605
1.78737848458691 72.59
1.82725306990817 72.59
1.84788694992274 72.515
1.88692324144669 72.515
1.91503433031275 72.515
1.93592447878205 72.51
1.97554450414141 72.51
1.99610350220973 72.415
2.03533590391305 72.415
2.06311300157575 72.415
2.08384133918123 72.415
2.12483331163529 72.415
2.1455307559707 72.345
2.18490272166373 72.345
2.21289256406301 72.345
2.23368402108831 72.32
2.2745260418433 72.32
2.29585157336327 72.255
2.33562262073684 72.255
2.36375141082853 72.255
2.38462496303128 72.24
2.42453877622318 72.24
2.44509492801903 72.12
2.48436521194975 72.12
2.5125709352731 72.12
2.5335335898443 72.105
2.57345113364061 72.105
2.59399793157495 72.065
2.6332530195833 72.065
2.66109378002854 72.065
2.68175831627576 72.05
2.72166196056208 72.05
2.74213211153304 71.975
2.78126138372239 71.975
2.8090145333116 71.975
2.82985229468639 71.975
2.86993203042977 71.975
2.89046406093248 71.965
2.92966463944967 71.965
2.95734694402234 71.965
2.97806183066557 71.965
3.01849466282885 71.965
3.03990331243392 71.925
3.07944147782989 71.925
3.10755645410974 71.925
3.12852757274041 71.925
3.16863673941528 71.925
3.18918316902993 71.88
3.22884444619488 71.88
3.25677665567907 71.88
3.27758329221646 71.88
3.3176546660088 71.88
3.33856537109303 71.84
3.37805072272711 71.84
3.40630048942758 71.84
3.42723920321875 71.84
3.46727967097746 71.84
3.48816650821626 71.805
3.52773086077771 71.805
3.55643409648444 71.805
3.57730722748992 71.795
3.61741566481887 71.795
3.63810618036798 71.77
3.67664559950959 71.77
3.70391168459581 71.77
3.72438645232649 71.77
3.76306419374637 71.77
3.78348069220763 71.745
3.8179219020029 71.745
3.84216049949985 71.745
3.86012487678612 71.745
3.89462114490536 71.745
3.91247563711102 71.715
3.93884139388217 71.715
3.95733046262777 71.715
3.97108013396626 71.715
3.99689136355962 71.715
4.01010232358671 71.71
4.03183426124369 71.71
4.04738501952641 71.71
4.0590238871731 71.71
4.08127029085951 71.71
4.09276458026324 71.71
4.11111584456617 71.71
4.12417414337409 71.71
4.13388923168147 71.7
4.15281278091214 71.7
4.1625320264913 71.7
4.1792516236152 71.7
4.19114053839587 71.7
4.19991996340142 71.7
4.21679707189491 71.7
4.2255027144432 71.69
4.24000616297307 71.69
4.25022672812135 71.69
4.25780792529163 71.69
4.27220581740897 71.69
4.27964927942446 71.67
4.292649773921 71.67
4.3020274412758 71.67
4.30893905147833 71.665
4.32213386013989 71.665
4.32887446424983 71.655
4.34093230101711 71.655
4.34961303822538 71.655
4.35600219970731 71.655
4.36835919593987 71.655
4.37470966015032 71.655
4.38608808566022 71.655
4.39395191714046 71.655
4.39986574929536 71.655
4.41138100192438 71.655
4.41730147842066 71.655
4.42755536920767 71.655
4.43509418043945 71.655
4.44068193661765 71.655
4.45048907567438 71.655
4.45583263771965 71.645
4.46473477153056 71.645
4.47098614680347 71.645
4.47568709126998 71.645
4.48444090894006 71.645
4.48898235274597 71.645
4.49576438005778 71.645
4.50061265969485 71.645
4.50419555697049 71.645
4.51087069316672 71.645
4.51436531952827 71.645
4.52003278409938 71.645
4.52413748644784 71.645
4.52724813588298 71.645
4.53312025856772 71.645
4.53626943023131 71.635
4.54091485323919 71.635
4.54419114442763 71.635
4.54665807497165 71.635
4.55141118590554 71.635
4.5538528555534 71.635
4.55809917135362 71.635
4.56112187990371 71.635
4.5634023476916 71.635
4.56775594298411 71.635
4.57002265530783 71.635
4.57393529170008 71.635
4.57668309518365 71.635
4.57873193364556 71.635
4.58278784092819 71.635
4.58484414242229 71.635
4.58792387123689 71.635
4.59016162123388 71.635
4.59184390319683 71.635
4.59494264423636 71.635
4.59656851676201 71.635
4.59942462082217 71.635
4.60146269049148 71.635
4.60301505588795 71.635
4.60602180817254 71.635
4.60757638348717 71.635
4.6101795613203 71.635
4.61202423397709 71.635
4.61346739440803 71.635
4.61587988474876 71.635
4.61711400187369 71.635
4.61927965055578 71.635
4.62083550404915 71.635
4.62202149892987 71.635
4.62429796625535 71.635
4.62544861338576 71.635
4.62711216574532 71.635
4.6282160490383 71.635
4.62902922051519 71.635
4.63072710501086 71.635
4.63158864306284 71.635
4.63283770083908 71.635
4.63368656285886 71.635
4.63432862066463 71.635
4.63556420559824 71.635
4.63620822899128 71.635
4.63721114163424 71.635
4.6379567629148 71.635
4.63853077092229 71.635
4.6394934765477 71.635
4.640005347929 71.635
4.64095261330115 71.635
4.64164234298003 71.635
4.64215019201855 71.635
4.64294894432647 71.635
4.64334614429453 71.635
4.64412441292211 71.635
4.6446701824787 71.635
4.64507342872449 71.635
4.64587714423144 71.635
4.6463151182835 71.635
4.6470974438978 71.635
4.64764920503117 71.635
4.64806105208879 71.635
4.6488651012715 71.635
4.64926641475058 71.635
4.64984105911239 71.635
4.65024184563904 71.635
4.65054415770209 71.635
4.65094663266456 71.635
4.65114452427135 71.635
4.65151662379936 71.635
4.65178096976022 71.635
4.65198104211134 71.635
4.65236923283384 71.635
4.6525669968051 71.635
4.65294688757096 71.635
4.65321908676408 71.635
4.65341795204769 71.635
4.65381928376042 71.635
4.65401987394714 71.635
4.65440516916632 71.635
4.65466915410095 71.635
4.65487691558149 71.635
4.65527659532569 71.635
4.65548323908361 71.635
4.65579831652373 71.635
4.65593262739941 71.635
4.65603291611099 71.635
4.65621851094548 71.635
4.656317579826 71.635
4.65649812211658 71.635
4.65663420530286 71.635
4.65673666928866 71.635
4.65692928043096 71.635
4.65702725711595 71.635
4.65722355327854 71.635
4.65736344912064 71.635
4.65746452552583 71.635
4.65766564995844 71.635
4.65776286265359 71.635
};
\addlegendentry{Random}

\addplot [color=plotcolor2, style=\linestyleB] table[]{
0.0510080928223974 134.735
0.0869022745746728 134.735
0.109703807149268 126.865
0.149962615551187 126.865
0.170823103345036 103.875
0.210368058951573 103.875
0.242410499226346 103.875
0.264213285546588 100.935
0.304182773359765 100.935
0.32499851578104 86.77
0.364330741237912 86.77
0.394151376841006 86.77
0.415646322656073 85.515
0.455541231201565 85.515
0.476308583344 78.005
0.515518752760119 78.005
0.544283536146202 78.005
0.565302603801424 77.585
0.605122741625923 77.585
0.625844912611594 74.485
0.665751798840559 74.485
0.694288415470448 74.485
0.715458999271748 74.33
0.755591621711607 74.33
0.776963032737056 73.085
0.817633607968541 73.085
0.846381042038406 73.085
0.867408765206407 73.01
0.907813420913192 73.01
0.928459971489868 72.585
0.967822588791484 72.585
0.995717464881082 72.585
1.01660483636577 72.555
1.05650792634919 72.555
1.07746910399488 72.33
1.11697474967936 72.33
1.1454844880885 72.33
1.1666760861693 72.325
1.20704010454445 72.325
1.22770940342056 72.125
1.2671626218965 72.125
1.29513300352311 72.125
1.31600510979756 72.105
1.35621717077234 72.105
1.37725475834399 72.025
1.41707781537556 72.025
1.44525998629559 72.025
1.46653027350837 72.025
1.50746966923798 72.025
1.52870680950177 71.95
1.56876067352179 71.95
1.59711863686705 71.95
1.61811454778912 71.95
1.65793034418199 71.95
1.67851675654032 71.84
1.71826569379213 71.84
1.74613902682644 71.84
1.7671836671732 71.84
1.80760116301501 71.84
1.82860202254919 71.775
1.86815799127921 71.775
1.89646849061496 71.775
1.91765235033548 71.775
1.9579825398232 71.775
1.97874689253051 71.735
2.01828476254138 71.735
2.04668346715008 71.735
2.06769034378814 71.735
2.10772987798407 71.735
2.12831280589224 71.695
2.16767321692243 71.695
2.1955755724545 71.695
2.21656422456123 71.69
2.25624240067129 71.69
2.27719699414663 71.64
2.31697409691476 71.64
2.34529600869089 71.64
2.36624342175559 71.64
2.40622081750197 71.64
2.42690388460581 71.6
2.46660428206287 71.6
2.49450031963585 71.6
2.51533111575704 71.6
2.55585848209255 71.6
2.57642320005631 71.58
2.61594751334977 71.58
2.64390644789206 71.58
2.66487927683893 71.58
2.70491206093831 71.58
2.72548617287759 71.555
2.76497984452533 71.555
2.79302369499045 71.555
2.81401498003598 71.55
2.85403034917801 71.55
2.87480318235508 71.52
2.91433566796964 71.52
2.94242003907106 71.52
2.96331716373415 71.515
3.00312523635367 71.515
3.02433282899063 71.51
3.06405741885206 71.51
3.09230378680052 71.51
3.11337663460097 71.51
3.15383183398045 71.51
3.17502407206229 71.5
3.214287868461 71.5
3.24236352023187 71.5
3.26336405885384 71.5
3.30353297719372 71.5
3.3242254893273 71.48
3.36322367617333 71.48
3.39101518024144 71.48
3.41193144476499 71.48
3.45189268367574 71.48
3.47277629720555 71.48
3.51135233680611 71.48
3.53895052408975 71.48
3.559442353045 71.48
3.59882931956036 71.48
3.61911920721556 71.475
3.65083027099212 71.475
3.6735533186516 71.475
3.69049985176044 71.475
3.72274273452939 71.475
3.73968784723704 71.45
3.76472115464212 71.45
3.78236859283279 71.45
3.7956931484879 71.45
3.81960928209387 71.45
3.83209502793942 71.44
3.8487178098617 71.44
3.86056967933941 71.44
3.86950682357815 71.44
3.88658523154728 71.44
3.89537755503553 71.44
3.90856509394158 71.44
3.91777026186801 71.44
3.92469473773267 71.44
3.93822549224468 71.44
3.94506734780283 71.44
3.95680538279174 71.44
3.96501686065445 71.44
3.97105383739889 71.44
3.98253092700232 71.44
3.98837679141035 71.43
3.99852536464561 71.43
4.00584853746085 71.43
4.0112423277367 71.43
4.02148292254747 71.43
4.02688609585866 71.43
4.03563242679282 71.43
4.04182139922828 71.43
4.04651753548177 71.43
4.05551617634563 71.43
4.06028236992207 71.43
4.06726337191938 71.43
4.07227462208029 71.43
4.0760721148078 71.43
4.08335995081291 71.43
4.08723152867624 71.41
4.09297857218136 71.41
4.09721592249679 71.41
4.10037497132759 71.41
4.10633340845596 71.41
4.10946471843783 71.4
4.11503390546875 71.4
4.11892399230832 71.4
4.12182287068369 71.4
4.1273987275831 71.4
4.13032518252794 71.4
4.13534841219569 71.4
4.13886526355098 71.4
4.14150914435689 71.4
4.14656872827233 71.4
4.14930236056456 71.4
4.15347874303064 71.4
4.15628143526526 71.4
4.15840851066468 71.4
4.16247485306415 71.4
4.1645424958874 71.4
4.16768960167866 71.4
4.17000027897482 71.4
4.17177178182492 71.4
4.17500084604128 71.4
4.17667125242918 71.4
4.17961489623413 71.4
4.18167416434279 71.4
4.1832380534049 71.4
4.18635732188552 71.4
4.18792200046459 71.4
4.19054350493749 71.4
4.19237583341447 71.4
4.19375611693312 71.4
4.19639333334063 71.4
4.19778896594429 71.4
4.20015224002191 71.4
4.20188798486461 71.4
4.20311344100667 71.4
4.20539343289635 71.4
4.20662013986667 71.4
4.20890082645834 71.4
4.21054839337856 71.4
4.21174340297493 71.4
4.21398919583412 71.4
4.21513551611982 71.4
4.21696306738099 71.4
4.21824889148536 71.4
4.21919025644897 71.4
4.22103056944048 71.4
4.22196771149118 71.4
4.22350536780368 71.4
4.22461180198404 71.4
4.22544149132279 71.4
4.22707834779809 71.4
4.22799122997991 71.4
4.22955271949392 71.4
4.23068675577086 71.4
4.23153521300366 71.4
4.23316856132505 71.4
4.23401017729306 71.4
4.23537141029611 71.4
4.23634684665466 71.4
4.2371177379482 71.4
4.23851922902301 71.4
4.23927459436515 71.4
4.24026910858976 71.4
4.24097468513225 71.4
4.24152102904876 71.4
4.24257103191144 71.4
4.24310444087794 71.4
4.24390462817398 71.4
4.24445612127272 71.4
4.24486726998162 71.4
4.24568664636897 71.4
4.24610206904499 71.4
4.24690017405841 71.4
4.24743788981261 71.4
4.24786190053691 71.4
4.24866927918378 71.4
4.24909006437572 71.4
4.24988042191203 71.4
4.25044837270161 71.4
4.25085782568122 71.4
4.25173438278919 71.4
4.25216486828195 71.4
4.25275560748499 71.4
4.25317757239397 71.4
4.2534967852255 71.4
4.25412415035758 71.4
4.25444406883129 71.4
4.25507952970316 71.4
4.25549150986641 71.4
4.255827406486 71.4
4.25648245029599 71.4
4.25686300294328 71.4
4.25749436124431 71.4
4.25795714946814 71.4
4.25828082407338 71.4
4.25895613640665 71.4
4.25931252217668 71.4
4.25951863351206 71.4
4.25966130269205 71.4
4.25976282946835 71.4
4.259973435398 71.4
4.26007803636738 71.4
};
\addlegendentry{ZSLA}

\addplot [color=plotcolor3, style=\linestyleC] table[]{
0.528269294937647 56.365
0.555999739988177 56.365
0.577871808701316 56.18
0.616376002258055 56.18
0.637492599873896 54.59
0.67546031756451 54.59
0.702848963763542 54.59
0.724070804173172 54.555
0.763544788950992 54.555
0.78417046200053 54.255
0.823289432178669 54.255
0.850598867763384 54.255
0.871695293494271 54.25
0.91218236164399 54.25
0.932997850617555 54.215
0.972767743801197 54.215
1.00031350205293 54.215
1.02164300303809 54.21
1.06188518235289 54.21
1.08307071045222 54.195
1.12236823545329 54.195
1.14993750234758 54.195
1.17096748101786 54.195
1.21110056476902 54.195
1.23252372835803 54.175
1.27285554428716 54.175
1.3002335146852 54.175
1.32131946336185 54.175
1.36228891179035 54.175
1.38346879110548 54.17
1.42362864549886 54.17
1.45147331888499 54.17
1.4726489862274 54.17
1.51399836371241 54.17
1.53496156346969 54.17
1.57489483925763 54.17
1.60243430143065 54.17
1.62351150160292 54.17
1.66366571268219 54.17
1.68464932865531 54.14
1.72428938367717 54.14
1.7517476951757 54.14
1.77282220770824 54.14
1.81307964311645 54.14
1.83384923070415 54.14
1.87344924281129 54.14
1.90059463938038 54.14
1.92159127417827 54.14
1.9628276668625 54.14
1.98408977810744 54.12
2.02398998535109 54.12
2.05112768732065 54.12
2.07212471049525 54.12
2.11204975743778 54.12
2.13262194025709 54.12
2.17205853146927 54.12
2.19917641711002 54.12
2.220301382147 54.12
2.26039369181415 54.12
2.2808458552912 54.12
2.32044057781702 54.12
2.34745971366608 54.12
2.36840943328728 54.12
2.40890702042862 54.12
2.42986556878224 54.12
2.46964564545838 54.12
2.49699709501518 54.12
2.51817858253807 54.12
2.558798389823 54.12
2.57973428651552 54.12
2.62006434654434 54.12
2.64700497924464 54.12
2.66796051539959 54.12
2.70790420624736 54.12
2.72886546047437 54.12
2.76829102032942 54.12
2.79525073818925 54.12
2.8163418743388 54.12
2.8571218678695 54.12
2.87781826012341 54.12
2.91744256127326 54.12
2.94478544466216 54.12
2.96585747377181 54.12
3.00561248282846 54.12
3.02646570086679 54.12
3.06591617700205 54.12
3.09332261377158 54.12
3.1143053745866 54.12
3.15446737690373 54.12
3.17526022150965 54.12
3.21447002260608 54.12
3.24166609230675 54.12
3.26257283521924 54.12
3.30260057224481 54.12
3.32325283177673 54.12
3.36242931088115 54.12
3.3896092802705 54.12
3.41047777263589 54.12
3.45087273168858 54.12
3.472103092682 54.12
3.51195033482448 54.12
3.53447152505868 54.12
3.55198682035444 54.12
3.58445672298311 54.12
3.60095263882826 54.12
3.61358165738756 54.12
3.62203133415968 54.12
3.62846815657453 54.12
3.64040113478933 54.12
3.64659458905551 54.12
3.64956570367477 54.12
3.65153247248634 54.12
3.65304679337593 54.12
3.65579916438837 54.12
3.65720275597945 54.12
3.65917269379906 54.12
3.66048806188354 54.12
3.66149918696163 54.12
3.66347929368682 54.12
3.66450955680201 54.12
3.6658561062482 54.12
3.6667901175815 54.12
3.66749293172629 54.12
3.66884051866706 54.12
3.66954559925432 54.12
3.67050992585105 54.12
3.67115829249908 54.12
3.67166816547254 54.12
3.67265444173492 54.12
3.67315072334706 54.12
3.67391468767766 54.12
3.67444048591943 54.12
3.6748368817836 54.12
3.6755723232184 54.12
3.67596882483773 54.12
3.67677738320163 54.12
3.67729334439022 54.12
3.67769612567797 54.12
3.67849905896142 54.12
3.67893723540697 54.12
3.67934751619421 54.12
3.67960202692526 54.12
3.67980110918928 54.12
3.68018014297368 54.12
3.68037630056054 54.12
3.6807422717594 54.12
3.68099696300356 54.12
3.68119545632075 54.12
3.68153622861786 54.12
3.68174463009945 54.12
};
\addlegendentry{SSLA}

\end{axis}
\end{tikzpicture}
\caption{\label{fig:speedup-M2M}In a three-color graph coloring experiment, when using an SSLA for initialization, the MGM-2 algorithm has a great benefit in speed and solution quality.}
\end{center}
\end{figure}
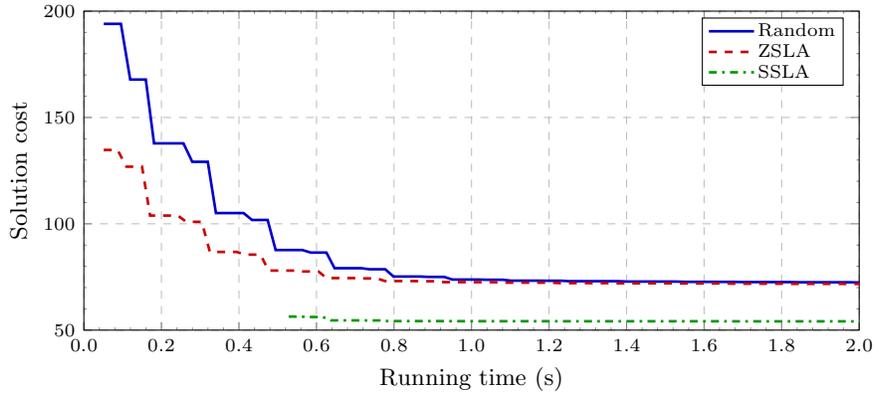

\begin{table}[t]
\centering
\caption{\label{tab:graphcoloring}Graph coloring experiment results}
\begin{tabular}{lccccc}
\toprule
Algorithm & I & S & M$^*$ & E$^*$ & T \\ \midrule
Random\_DSA & 157 & 49 & 10.4 & 362.9 & 3.5 \\
ZSLA\_DSA & 164 & 49 & 10.7 & 381.6 & 3.7 \\
SSLA\_DSA & 115 & 47 & \textbf{14.5} & 381.3 & 3.1 \\ \midrule
Random\_MGM2 & 55 & 72 & 26.3 & 120.2 & 1.7 \\
ZSLA\_MGM2 & 42 & 71 & 21.0 & 94.4 & 1.3 \\
SSLA\_MGM2 & 7 & 54 & 12.2 & 121.0 & 0.7 \\ \bottomrule
\end{tabular}

* = $\times 10^3$
\end{table}

\subsubsection{Experiment 1: Symmetric DCOP}
We use a (symmetric) graph-coloring problem with three colors, which have to be assigned to 200 variables. The constraints between the variables are chosen as the nodes were connected via a Delaunay triangulation, when the variables are points chosen randomly on a two-dimensional plane. The results of \mbox{MGM-2} ($p=0.5$) is shown in Figure~\ref{fig:speedup-M2M}, of which the numeric results are shown in Table~\ref{tab:graphcoloring} together with DSA (variant C, with $p=0.5$).

\subsubsection{Experiment 2: Asymmetric DCOP}
An asymmetric problem is chosen where the constraints are created using a scale-free graph generation method~\cite{Albert2002}, and are assigned semi-random asymmetric costs such that there is a high probability that a conflict of interests occurs. This problem is created specifically to benchmark asymmetric problems, and is described in more detail in~\cite[Section 5.2]{Grinshpoun2013}. The result of the ACLS ($p=0.5$) algorithm is shown in Figure~\ref{fig:speedup-ACLS}, and the results of all algorithms is presented in Table~\ref{tab:semirandom}.

\begin{figure}[t]
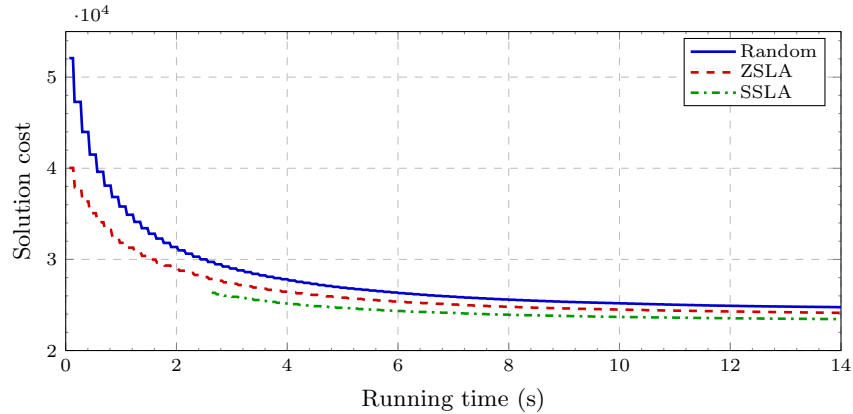

\begin{center}

\caption{\label{fig:speedup-ACLS}In a semi-random experiment (Section~\ref{sec:experiment_init}, Experiment 2), when using an SSLA for initialization, the ACLS algorithm shows faster convergence to a better solution.}
\end{center}
\end{figure}

\begin{table}[t]
\centering
\caption{\label{tab:semirandom}Semi-randomized asymmetric experiment results}
\begin{tabular}{lrrrrr}
\toprule
Algorithm & I & S$^*$ & M$^*$ & E$^*$ & T \\ \midrule
Random\_ACLS & 95 & 26.4 & 151 & 1321 & 4.3 \\
ZSLA\_ACLS & 75 & 26.1 & 121 & 995 & 3.5 \\
SSLA\_ACLS & 48 & 24.2 & 107 & 10571 & \textbf{4.9} \\ \midrule
Random\_ACLSUB & 333 & 24.7 & 531 & 5747 & 14.8 \\
ZSLA\_ACLSUB & 299 & 24.1 & 477 & 5164 & 13.3 \\
SSLA\_ACLSUB & 207 & 23.5 & 358 & 12963 & 11.9 \\ \midrule
Random\_MCSMGM & 1154 & 21.7 & 1389 & 5557 & 55.1 \\
ZSLA\_MCSMGM & 1022 & 21.6 & 1232 & 4935 & 48.8 \\
SSLA\_MCSMGM & 781 & 22.0 & 970 & 13599 & 40.0 \\ \bottomrule
\end{tabular}

* = $\times 10^3$
\end{table}

\paragraph{Hybrid Solvers---Discussion of Initial of Results:}
Based on the results from the first two experiments, we see that local search DCOP solvers reduced their execution time and communication overhead using initialization methods other than random, but surprisingly also found a final solution with a lower cost. The only two exceptions (marked in bold in Table~\ref{tab:graphcoloring} and~\ref{tab:semirandom}) are (i) when using the SSLA with DSA, in which case the messages of the SSLA increases the very low communication overhead of DSA, and (ii) when using an SSLA with ACLS, in which case the added run time of the SSLA increases the convergence time. The speed performance gain can be easily explained: a better initialization will reduce the amount of variable ``tweaking'' needed. However, how come the final solution is also \emph{lower} and solution result dependent on the initialization method?

\subsection{Why Do Only Some Hybrid DCOP Solvers Work?}

We introduce three hypotheses to explain this phenomenon:
\begin{enumerate}
\item \textbf{Hypothesis 1:} A lower initial solution will always lead to a lower final solution;
\item \textbf{Hypothesis 2:} Using initialization methods other than random increases the explored portion of the solution space;
\item \textbf{Hypothesis 3:} The initialization method itself finds a starting point in the search space, from which a lower local minimum is reachable.
\end{enumerate}
Let us experimentally verify these three hypotheses in detail.

\subsubsection{Verifying Hypothesis 1: Solution Cost Correlation}

The SSLA algorithm finds a lower initial cost than the ZSLA initializer, which in turn finds a lower cost than a random assignment. Hence the first hypothesis is that a lower initial costs will (on average) lead to lower final costs. The initial state is known to be of great influence on the final solution, and a correlation between the initial cost and the final cost could explain why ZSLA or SSLA initialization methods lead to better final solutions.

To test this hypothesis we performed an experiment by repeatedly invoking the algorithms on the exact same problem setup, but with different random initializations. We gather information on the cost at initialization and of the eventual outcome. In Figure~\ref{fig:best-repeated-DSA} we show the minimum, average, and maximum results of 200 instantiations of the algorithms solving the exact same three color graph coloring problem with a Delaunay graph of size $n=100$. For this small experiment we only compare the DSA algorithm instantiated with a randomized approach with an algorithm that uses CoCoA\_WPT for initialization.

\begin{figure}[t]
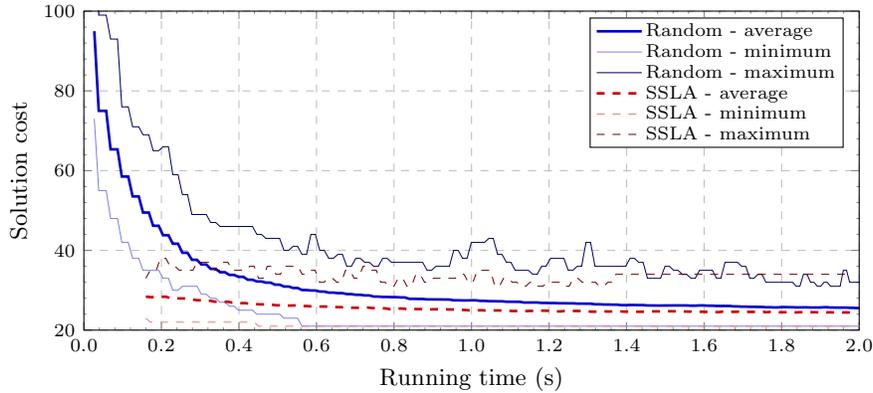

\begin{center}

\caption{\label{fig:best-repeated-DSA}Starting the DSA algorithm from a different starting point will lead to a different outcome. This graph shows the minimum, average and maximum solution costs during the experiments.}
\end{center}
\end{figure}

From this experiment we see that for some iterations we \emph{do} find a solution with the random strategy which is as good as the SSLA-initialized solution, however on average the final solution is worse. The spread of the initial and the final solution corresponds with the statistical spread of the random assignments, and some random initialization lead to better final solutions than others. If we look at the correlation between the initial cost and the final cost of the individual runs, then we find that there is no correlation between the cost of the initial random assignment and the final minimal cost. The Pearson correlation coefficient between the solution cost at the beginning and the end is 0.15. With these results we \emph{reject hypothesis 1}.

\subsubsection{Verifying Hypothesis 2: Increased Solution Space Exploration}

A DCOP problem is generally a matter of solution space exploration. The solvers that are capable of effectively searching a larger portion of the solution space, will reasonably find a better final solution than solvers that cannot. Since local search algorithms search only a small fraction of the search space, the increase in search space exploration by a SSLA may be of large influence. Put differently, a better overview of the trends in the solution space may lead to insights as to where the best optimum lies. If we can show that the solvers using SSLA explore a larger part of the solution space than randomly initialized solvers, this may explain why they find solutions with the lower final cost.

To verify this, we construct an experiment in which we captured the CPA \textbf{every} time a constraint check is performed, so that we can store every explored value assignment. We did observe that SSLA searches a slightly larger portion of the solution space than ZSLA, which in turn searches a larger part of the solution space than random. However, the SSLA algorithm sometimes searches a \textbf{smaller} part of the solution space, largely because its successor algorithm converges so quickly. Therefore, with these results we \emph{reject hypothesis 2} as well, also because the differences are so marginally small that they cannot explain the significant effect on the results.

\subsubsection{Verifying Hypothesis 3: Selection of Starting Point}
The initial assignment determines the area of the solution space that is reachable through local search. It is possible that SSLA is capable of finding initial assignments that have relatively good local minima. To explain what we mean by this, let us sketch the following example.

\begin{definition}
A \textbf{bridge edge} is defined as an edge that, when removed from the graph, the graph will no longer be connected.
\end{definition}

Suppose we have a graph with high modularity, meaning it consists of clusters of densely connected nodes, which are connected through relatively low number of \textit{bridge} edges. The nodes on these bridges could initially induce a high performance penalty; and it may be impossible for a local search algorithm to escape from these expensive assignments because of the many low cost constraints around it on the surrounding nodes.

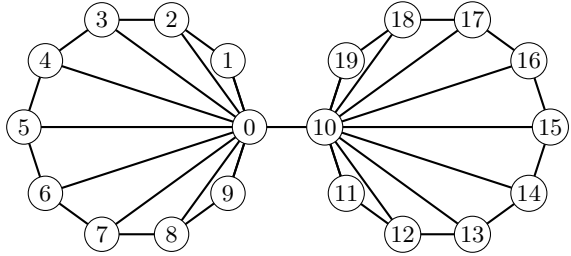
\begin{figure}[t]
\centering
\begin{tikzpicture}
\tikzset{VertexStyle/.style ={shape = circle,minimum size = 13pt,inner sep = 1pt,font = \small,draw}}
\begin{scope}[shift={(-2,0)}]
\grEmptyCycle[RA=1.5,prefix=]{10}
\EdgeInGraphFromOneToComp{}{10}{0}
\EdgeInGraphLoop{}{10}
\end{scope}
\begin{scope}[rotate=180,shift={(-2,0)}]
\grEmptyCycle[RA=1.5,prefix=1]{10}
\EdgeInGraphFromOneToComp{1}{10}{0}
\EdgeInGraphLoop{1}{10}
\end{scope}
\Edge(0)(10)
\end{tikzpicture}
\caption{\label{fig:bridge_topology}An example graph in which two dense clusters of nodes are connected by a single bridge.}
\end{figure}

As a minimal example suppose we have a three color graph coloring problem with a graph as shown in Figure~\ref{fig:bridge_topology}. As we see node~0 and 10 have a constraint with many other nodes, and are connected to one another. If through some unfortunate random assignment, they are both given the same initial color (for example red) and the nodes around it are mostly other colors, then no local search algorithm will change that initial assigned color. This is confirmed through experiments in which we use the graph as depicted in Figure~\ref{fig:bridge_topology}, letting the algorithms solve the graph-coloring problem. One group of agents is hardwired to intialize with an initialization in which $\var{0} = \var{10}$, and $\forall_{i \neq 10}{\var{i} \neq \var{0}}$. As we can see in Figure~\ref{fig:manualgraph}, the local search algorithm is unable to find a solution in which this constraint is resolved. However, we can guarantee that an SSLA algorithm will never assign the same color to endpoint vertices of a bridge, and will thus lead to better solutions.

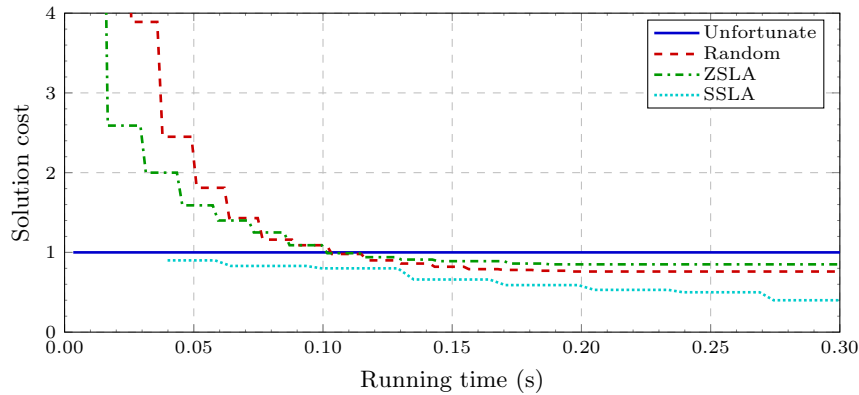
\begin{figure}[t]
\begin{center}
\begin{tikzpicture}

\begin{axis}[
xmin=0,
xmax=0.3,
xlabel={Running time (s)},
ymin=0,
ymax=4,
ylabel={Solution cost},
x tick label style={
/pgf/number format/.cd,
fixed,
fixed zerofill,
precision=2,
/tikz/.cd
},
legend style={legend cell align=left, align=left, draw=black}
]
\addplot [color=plotcolor1, style=\linestyleA] table[]{
0.00336255212540365 1
0.00672471040429439 1
0.00989008035562881 1
0.0118974285104761 1
0.0148807757317614 1
0.0180755528893929 1
0.0212799683464099 1
0.0232631934023416 1
0.0263495059596636 1
0.0295575900240137 1
0.0328070425591271 1
0.0347523270676486 1
0.0379263069413625 1
0.0413562104672004 1
0.0448109190298977 1
0.0467843125099146 1
0.0500275036148181 1
0.0534626511340409 1
0.0569076083488167 1
0.0588462558142503 1
0.062168292803002 1
0.0657843982079984 1
0.0693163809152191 1
0.0712661217240996 1
0.0747933126320799 1
0.0785279585440078 1
0.0821988232012063 1
0.0841606857305397 1
0.0877534156156495 1
0.0917234869273955 1
0.0954362305971333 1
0.0974408035927554 1
0.10096665250521 1
0.104757877386099 1
0.108527309426607 1
0.110477557130536 1
0.114045193887356 1
0.11781051241984 1
0.121506185760625 1
0.12346156805613 1
0.127074296591951 1
0.130900613926486 1
0.134582396884236 1
0.136563630827242 1
0.140131457213864 1
0.143797034116954 1
0.147442765531866 1
0.149387218586638 1
0.153897479564654 1
0.157867649338028 1
0.161523059166322 1
0.163517206169533 1
0.167028023273412 1
0.170839162930291 1
0.174549313777152 1
0.17650755875333 1
0.180063963591078 1
0.183711146403324 1
0.187381952712891 1
0.189326865255262 1
0.192754606272006 1
0.197463289311261 1
0.201118389167762 1
0.203090845438947 1
0.206620753158521 1
0.210342697520408 1
0.214030304301132 1
0.215991751103591 1
0.219615062441082 1
0.223378798294061 1
0.227015223261742 1
0.229027403329826 1
0.232538001630087 1
0.236359563632651 1
0.240162246529683 1
0.242091857405682 1
0.245558997660625 1
0.249237545971552 1
0.252929864323523 1
0.254883733227335 1
0.258475919750126 1
0.26225993869852 1
0.265937400284809 1
0.267929691103992 1
0.271495139824629 1
0.275244215835182 1
0.278897029193873 1
0.280852298440842 1
0.284289294850639 1
0.287924669560952 1
0.291647485490585 1
0.293585308795723 1
0.297108535711485 1
0.300667296334857 1
0.304199330096255 1
0.306158413819637 1
0.309632220291483 1
0.313310075724286 1
0.316961113126941 1
0.319151898941902 1
0.322607734343234 1
};
\addlegendentry{Unfortunate}

\addplot [color=plotcolor2, style=\linestyleB] table[]{
0.00498889024628171 11.05
0.00849596216156094 11.05
0.0117600052512868 11.05
0.013723173308876 6.47
0.0169928214179568 6.47
0.0204747527974954 6.47
0.0238712814780914 6.47
0.0258180976119408 3.89
0.0290457682468542 3.89
0.0325231339242247 3.89
0.0358647538368126 3.89
0.0378992701075967 2.45
0.0422318661942939 2.45
0.0457352403284242 2.45
0.049099947669468 2.45
0.0511259014253232 1.81
0.0544919617020733 1.81
0.0581566415103284 1.81
0.0618153151592617 1.81
0.0637432594810343 1.43
0.0672598748807976 1.43
0.0710609313376012 1.43
0.0747472034162538 1.43
0.0766630734250242 1.16
0.0803684835268226 1.16
0.084212881333681 1.16
0.0880295931893727 1.16
0.0899909779974728 1.09
0.0937291211205663 1.09
0.097653994898229 1.09
0.10142410887668 1.09
0.103450981607733 0.98
0.107120413101231 0.98
0.111001726725221 0.98
0.114842346522937 0.98
0.116820334878938 0.9
0.120490805689623 0.9
0.124317487696855 0.9
0.128015214144925 0.9
0.130025993869852 0.86
0.133603079296255 0.86
0.137451966224015 0.86
0.141260297901126 0.86
0.143189679033326 0.82
0.146827208959279 0.82
0.150749427919706 0.82
0.154472010458813 0.82
0.156417645053124 0.79
0.1600272848112 0.79
0.163988709733297 0.79
0.167690626270656 0.79
0.169662203680641 0.78
0.173272142470329 0.78
0.177085371701763 0.78
0.180841062875043 0.78
0.182797367792472 0.77
0.18648964238372 0.77
0.190443719150969 0.77
0.194283193876416 0.77
0.196265313974075 0.76
0.199947490778339 0.76
0.203809943530433 0.76
0.207612670188189 0.76
0.209609628817895 0.76
0.213272397741217 0.76
0.217098310289058 0.76
0.22083939267409 0.76
0.222821972259348 0.76
0.226480977760436 0.76
0.230297303063068 0.76
0.23413655533817 0.76
0.236121848088295 0.76
0.239641581439618 0.76
0.243362358848874 0.76
0.247255090375011 0.76
0.249266267593178 0.76
0.252845734332293 0.76
0.25665528766294 0.76
0.260427345346868 0.76
0.262466773758882 0.76
0.266066246442162 0.76
0.269900816319833 0.76
0.273607593944245 0.76
0.27562137492547 0.76
0.27917441748095 0.76
0.283145542696791 0.76
0.286725731487846 0.76
0.288663226587557 0.76
0.292214248856295 0.76
0.29583304189907 0.76
0.299417482773773 0.76
0.30137427270589 0.76
0.304884440692368 0.76
0.308551224662085 0.76
0.312150762986451 0.76
0.314124944159493 0.76
0.317658899746005 0.76
0.32133745170366 0.76
0.325100447271063 0.76
0.327508588953699 0.76
0.331153875467921 0.76
0.334856098330346 0.76
0.338534766983263 0.76
0.340493847059918 0.76
0.344043789897472 0.76
0.347713520422583 0.76
0.351229712802017 0.76
0.353179482784714 0.76
0.356622372305297 0.76
0.360135220636098 0.76
0.363566852710521 0.76
0.36537898063041 0.76
0.367985372978118 0.76
0.370746805923014 0.76
0.373502557267289 0.76
0.374986392238306 0.76
0.377134401945894 0.76
0.379445938184331 0.76
0.381715923615657 0.76
0.382947948442574 0.76
0.384683546150242 0.76
0.386517729474853 0.76
0.388391177108766 0.76
0.389380975390063 0.76
0.390641955229133 0.76
0.391918849384706 0.76
0.393194817271628 0.76
0.393888504969577 0.76
0.394844403276949 0.76
0.395884683199711 0.76
0.396859792464768 0.76
0.397397586231418 0.76
0.398122289342258 0.76
0.398923431497146 0.76
0.399678902043444 0.76
0.400047217820826 0.76
0.400512686051452 0.76
0.401028406179744 0.76
0.401459354492859 0.76
0.401704210328625 0.76
0.402030231366593 0.76
0.402370274069766 0.76
0.402722234276681 0.76
0.402915160720375 0.76
0.403110435656238 0.76
0.403315775558542 0.76
0.40352769780303 0.76
0.403644024746689 0.76
0.403739992013668 0.76
0.403855797475371 0.76
0.403959732840782 0.76
0.40401591796323 0.76
0.404082656713533 0.76
0.404153800710018 0.76
0.404217662192741 0.76
0.404260051747056 0.76
0.404297478105963 0.76
0.40433721648977 0.76
0.404377056981932 0.76
0.404393751698007 0.76
};
\addlegendentry{Random}

\addplot [color=plotcolor3, style=\linestyleC] table[]{0.00807821500008206 7.41
0.0113594050000274 7.41
0.0146470971141626 7.41
0.0167163995135266 2.59
0.0208858665626134 2.59
0.0251599764421438 2.59
0.0293464773529138 2.59
0.0314461679281303 2
0.0353317883366731 2
0.0394337909367895 2
0.0434642739275432 2
0.0455680707173294 1.59
0.0493067681429225 1.59
0.0533125518518991 1.59
0.0573343301053722 1.59
0.0594153567319491 1.4
0.0633251148263155 1.4
0.0672824153002076 1.4
0.0711822542972119 1.4
0.0732395407312052 1.25
0.0770166965394384 1.25
0.0809915815307865 1.25
0.0849312026723215 1.25
0.0870578425598565 1.09
0.091113068593111 1.09
0.0953816901485495 1.09
0.0995562334415804 1.09
0.101641311581823 0.99
0.105862193834479 0.99
0.110243247629172 0.99
0.114175079361896 0.99
0.116252466554955 0.94
0.120235243063469 0.94
0.124220882252656 0.94
0.128288718667778 0.94
0.130352313939431 0.91
0.134075585709936 0.91
0.138097210800876 0.91
0.142324445651916 0.91
0.144410278664642 0.89
0.148275521162868 0.89
0.152431247344727 0.89
0.156503237381869 0.89
0.158598424249276 0.89
0.162377450828445 0.89
0.166268658022708 0.89
0.170132748155212 0.89
0.17217313565642 0.86
0.176052501928207 0.86
0.181389289927558 0.86
0.185566068664222 0.86
0.187601952457621 0.85
0.191357847847611 0.85
0.195501317380118 0.85
0.199393691527012 0.85
0.201523679109907 0.85
0.205334195176474 0.85
0.20927821062401 0.85
0.213238377425301 0.85
0.215397794094855 0.85
0.219182721078264 0.85
0.223083442583196 0.85
0.226830709817171 0.85
0.228838360650357 0.85
0.23260039348184 0.85
0.236473443622513 0.85
0.240302350133197 0.85
0.242296285626243 0.85
0.245996488201926 0.85
0.249873815953337 0.85
0.25386168693943 0.85
0.256019612097652 0.85
0.259716773303041 0.85
0.263634149410051 0.85
0.267673333491358 0.85
0.269730204198477 0.85
0.273475356330809 0.85
0.277358121352133 0.85
0.281485144146 0.85
0.283591887491179 0.85
0.287238636342916 0.85
0.291171449045196 0.85
0.294984145854492 0.85
0.297098766129929 0.85
0.300805835492499 0.85
0.30472091051479 0.85
0.308712198484056 0.85
0.310972027780766 0.85
0.314889345540144 0.85
0.318840027204583 0.85
0.322736044431721 0.85
0.324852867330249 0.85
0.32861479440665 0.85
0.332593672564032 0.85
0.33647297684146 0.85
0.338607902092674 0.85
0.342447238242496 0.85
0.346348492169566 0.85
0.350234564772253 0.85
0.352679126317152 0.85
0.357547123917606 0.85
0.361440187295897 0.85
0.365324232318389 0.85
0.36758387927875 0.85
0.370750449003258 0.85
0.374059871963416 0.85
0.377381949066165 0.85
0.379101431887345 0.85
0.381348975360889 0.85
0.383661638437961 0.85
0.386014116480106 0.85
0.387275450051692 0.85
0.388994940166327 0.85
0.390803979308471 0.85
0.392649489367056 0.85
0.393655431708656 0.85
0.395072983770242 0.85
0.396616398966518 0.85
0.398102874167863 0.85
0.39888558211791 0.85
0.400013602291603 0.85
0.401160384875565 0.85
0.402428814977837 0.85
0.403078074601094 0.85
0.403993494239083 0.85
0.404881552484606 0.85
0.405703043375994 0.85
0.406109948818187 0.85
0.406468950854884 0.85
0.40690410019747 0.85
0.407278144982924 0.85
0.40748709514493 0.85
0.407764848104705 0.85
0.408032255300062 0.85
0.408303509792374 0.85
0.408456730672803 0.85
0.408665680834809 0.85
0.408902419056337 0.85
0.409116602271546 0.85
0.409231269954434 0.85
0.409367679423525 0.85
0.409509103142202 0.85
0.409647077057164 0.85
0.409723297297593 0.85
0.409869965009655 0.85
0.410055040050179 0.85
0.410193112426769 0.85
0.410264475226872 0.85
0.410292179411674 0.85
0.410323140123661 0.85
0.4103577585028 0.85
0.410376258348726 0.85
};
\addlegendentry{ZSLA}

\addplot [color=plotcolor4, style=\linestyleD] table[]{
0.0398528436265241 0.9
0.0490106247390311 0.9
0.0582187963248285 0.9
0.0643233261067361 0.83
0.073583416144425 0.83
0.0834236019816314 0.83
0.0933965943216814 0.83
0.0994773073297389 0.8
0.108664780093247 0.8
0.118809493159652 0.8
0.128945636417674 0.8
0.135118268825772 0.66
0.144950187532935 0.66
0.154945173283349 0.66
0.164399462472444 0.66
0.17056101977073 0.59
0.1797313565642 0.59
0.189207919961636 0.59
0.198737674518678 0.59
0.205000869744383 0.53
0.214406978376732 0.53
0.22393030740085 0.53
0.233666207786856 0.53
0.239276456548337 0.5
0.2498092360654 0.5
0.259662466244983 0.5
0.268932603015479 0.5
0.27438198371007 0.4
0.283989723523395 0.4
0.293484677364948 0.4
0.303402640595 0.4
0.309260017103149 0.38
0.318434667974626 0.38
0.328142495856406 0.38
0.337865370133671 0.38
0.343601828468903 0.37
0.352731405795014 0.37
0.362725837242928 0.37
0.372489642383719 0.37
0.37829544687904 0.36
0.387367026659397 0.36
0.396882372998175 0.36
0.406634132999779 0.36
0.412361255713966 0.34
0.421042340323501 0.34
0.430053154692335 0.34
0.439697018983037 0.34
0.444910573137844 0.3
0.453731097646585 0.3
0.463743332415574 0.3
0.473151486861754 0.3
0.47899042187161 0.3
0.487776867716803 0.3
0.49685056624553 0.3
0.508250789060548 0.3
0.513636643771299 0.3
0.522810390254487 0.3
0.532680337030507 0.3
0.541582161670347 0.3
0.547495905637293 0.29
0.556340338088057 0.29
0.566158767552153 0.29
0.575239190645416 0.29
0.58058024896205 0.29
0.58954956357795 0.29
0.599298227508356 0.29
0.608895895061785 0.29
0.614470548121298 0.29
0.623177808937034 0.29
0.632628874419487 0.29
0.641905739401244 0.29
0.647252391797052 0.29
0.65577340332618 0.29
0.665228852174452 0.29
0.675190350031088 0.29
0.68094518422353 0.29
0.689978910977924 0.29
0.700037831145601 0.29
0.709285361855601 0.29
0.714736088192445 0.29
0.724571175905346 0.29
0.733646785319007 0.29
0.743080496027802 0.29
0.748952182292588 0.29
0.758177066828095 0.29
0.76767882181545 0.29
0.777131283994333 0.29
0.782731788701346 0.29
0.792278927205859 0.29
0.801612695715278 0.29
0.811094477579011 0.29
0.81676730417532 0.29
0.82892069280519 0.29
0.840353163626816 0.29
0.849681863185744 0.29
0.855422281866468 0.29
0.864215558031278 0.29
0.874086875976639 0.29
0.883093055355492 0.29
0.889572439496241 0.29
0.89847509558983 0.29
0.903193548210642 0.29
0.908123295838902 0.29
0.911591194613054 0.29
0.916023182243356 0.29
0.92049408409717 0.29
0.925266541097701 0.29
0.928097225387783 0.29
0.932201510838984 0.29
0.936369635892545 0.29
0.940292328927479 0.29
0.942978515308048 0.29
0.946136967418317 0.29
0.949417745338115 0.29
0.952603526020309 0.29
0.954826567135331 0.29
0.957928013609585 0.29
0.961615193723253 0.29
0.964885089809768 0.29
0.966556366547114 0.29
0.969143000928091 0.29
0.971582741499935 0.29
0.973611747566265 0.29
0.974981997932305 0.29
0.977278433074354 0.29
0.979287568125418 0.29
0.981360057764155 0.29
0.982675596285443 0.29
0.983834230732062 0.29
0.985022921502378 0.29
0.985933728030748 0.29
0.986596338321447 0.29
0.987511517275457 0.29
0.988400118883298 0.29
0.989378499262449 0.29
0.98997958926914 0.29
0.990561690768492 0.29
0.991221744703584 0.29
0.991799408136212 0.29
0.992245880566044 0.29
0.992952747535268 0.29
0.993638602063682 0.29
0.994317053736345 0.29
0.994679356060951 0.29
0.995073702175454 0.29
0.99557802628196 0.29
0.995987604775024 0.29
0.996347274162756 0.29
0.996437589002929 0.29
0.996540211546631 0.29
0.996632845705158 0.29
0.996709047711952 0.29
0.996896230560665 0.29
0.997040196048041 0.29
0.997093423674915 0.29
0.99713445300007 0.29
0.997268758307699 0.29
0.997381912598894 0.29
0.997491671787278 0.29
0.997592529315126 0.29
};
\addlegendentry{SSLA}

\end{axis}
\end{tikzpicture}
\caption{\label{fig:manualgraph}The results of the MCS-MGM algorithm trying to solve the graph coloring problem in graph from Figure~\ref{fig:bridge_topology}, with various initialization strategies. The ``\emph{unfortunate}'' strategy is hardwired to get a conflict on the bridge.}
\end{center}
\end{figure}

\begin{proposition}
A SSLA will never assign the same color to the endpoints of a bridge.
\end{proposition}

\begin{proof} (\emph{Sketch})
When an SSLA starts, any random agent is activated first. If either bridge endpoint is selected first, then logically they will not be activated simultaneously. If any other node is selected, then this node will execute the algorithm and select a random initial assignment. After that it activates all of its neighbors, which will execute the algorithm, until at some iteration the first bridge endpoint is selected. At this moment the other endpoint cannot be activated, or the edge would not have been a true bridge.

Because the algorithm is active in one bridge endpoint, but never in both at the same time, one must assign a value before the other. Moreover, when an endpoint of the bridge eventually has to assign a value, no nodes from the other component can be assigned a value yet, because the bridge is the only connecting edge. Therefore, when the second bridge endpoint is activated it will only have the first endpoint as a constraining value, and will thus always pick a different value.
\end{proof}

Although this exact order of events will not hold for pseudo-bridges that connect clusters within a graph, there will be an ordering in which the nodes will be activated, as long as the detour path between the vertices on the pseudo-bridge is longer than three. In many graphs with high modularity, the values of nodes in the bridging constraints will therefore be chosen with low costs, and the coloring within the clusters can simply be permuted. Therefore, the local search algorithms that continue from these solutions are generally of higher quality, than those from random initial assignments.

If this final hypothesis is true, then we expect some different results in the performance of different types of graphs, especially for various densities. We would expect the benefit of an SSLA to decrease with problem graphs of higher densities, since in these graphs bridges or pseudo-bridges occur less frequently.

\section{Graph Density}

\begin{sidewaysfigure}[p]
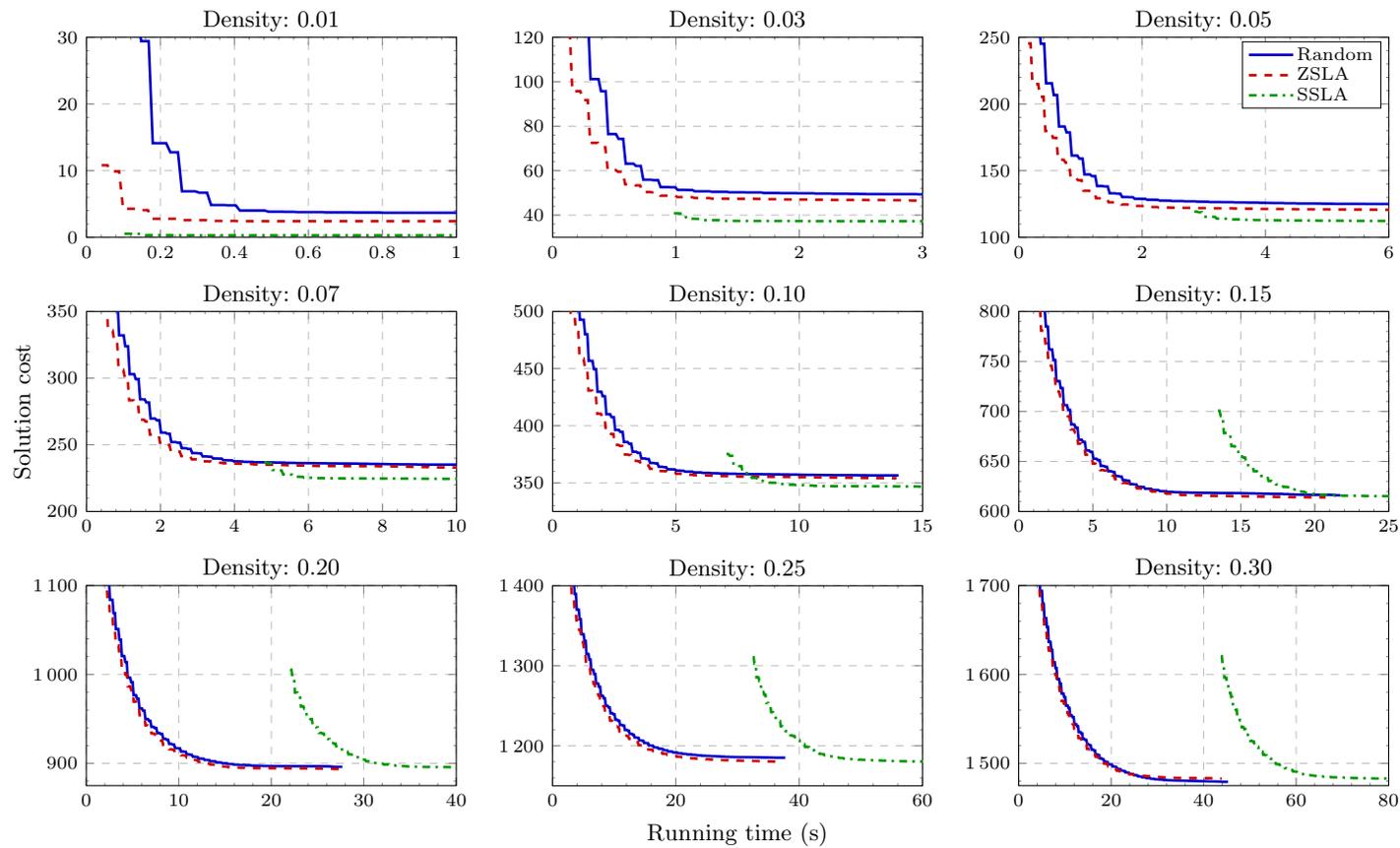

\begin{center}

\caption{\label{fig:densities}The solution cost versus the time of the MGM-2 algorithm when solving random graphs of different densities shows the impact of initialization methods. Up to a density of $0.1$ there is a clear improvement on the solution cost.}
\end{center}
\end{sidewaysfigure}

In our final experiment we use once more the graph coloring problem with $\domainsize{} = 3$, and instantiate randomly connected graphs with $n=200$ with nine varying densities between $0.01$ and $0.3$. We generate 50 graphs of every density, let the different solver combinations (initialization and iteration) solve the same graphs, and report the average performance of all instances. The convergence criteria were identical to the experiment described in Section~\ref{sec:first_results}.

In Figure~\ref{fig:densities} we show the averaged results for the MGM-2 solver, when solving the graphs with different densities. We can indeed conclude that for graphs with a low to medium density (up to $0.1$), there is a benefit using SSLA initialization for the final solution cost. The increase in convergence speed deteriorates much faster, since the complexity of the SSLA is exponential with the node degree, and this increases with graph density. Note, the time the SSLA initialization takes is shown as the starting point of the line.

For other local search algorithms (DSA, ACLS, MCS-MGM), similar results were found.

\section{Discussion}
In this article we studied the effect of different initialization strategies on the performance of different DCOP algorithms. Particularly, we introduced a new class of hybrid algorithms which combine the strategies of SSLA algorithms with iterative local search algorithms. We found that using this combination not only combines the fast convergence of the SSLA with the eventual better solution quality of the iterative approach, but that using the hybrid solver actually improves the quality of the final solution compared to using the iterative approach alone.

Two possible hypotheses that could explain this observation were rejected: (i) better initializations do not necessarily lead to lower final solution costs, and (ii) using SSLA does not significantly increase the searched solution space. Instead, we hypothesize that using an SSLA (such as CoCoA) selects an initialization that is in a region of the solution space which has a lower local minimum than the statistical expected local minimum. This is caused by a reduction of conflicting values assigned on bridge vertices. In our final experiment we show that the effect is most abundant on low density graphs, in which (pseudo-)bridges are more present, and the solution cost of the search space is less homogenous.

Our hybrid approach seems well suited for applications in which maximum performance in terms of both convergence speed and solution cost is required. In fact we may use it as a general strategy for initial value assignment instead of using random values in problems with low graph densities.

\end{document}